\documentclass[12pt]{amsart}

\usepackage{amsmath,amsthm,amsfonts,amssymb,eucal, bbm}
\usepackage{scalerel}
\usepackage{enumitem}
\usepackage{color}
\usepackage{dsfont}

\newcommand{\oo}{\"o}
\newcommand{\intt}{\textnormal{Int}}

\newcommand{\g}{\gamma}
\newcommand{\G}{\Gamma}

\newcommand{\s}{\sigma}

\newcommand{\om}{\omega}
\newcommand{\Om}{\Omega}

\newcommand{\real}{\mathbb{R}}
\newcommand{\naturals}{\mathbb{N}}
\newcommand{\integers}{\mathbb{Z}}
\newcommand{\complex}{\mathbb{C}}

\newcommand{\oq}{\ {\raise 7pt\hbox{${\scriptstyle\circ}$}}
	\kern -7pt{
		\hbox{$Q$}}}

\newcommand {\ba}{\mathbf a}

\newcommand {\bx}{\mathbf x}
\newcommand {\hbx}{\hat{\bx}}

\newcommand {\bh}{\mathbf h}

\newcommand {\bm}{\mathbf m}

\newcommand {\by}{\mathbf y}

\newcommand {\bn}{\mathbf n}

\newcommand {\bnu}{\boldsymbol\nu}

\newcommand\norm[1]{\left\lVert#1\right\rVert}

\DeclareMathOperator*{\esssup}{ess-sup}

\DeclareMathOperator {\diam} {{diam}}

\newtheorem{thm}{Theorem}[section]
\newtheorem{cor}[thm]{Corollary}
\newtheorem{lem}[thm]{Lemma}
\newtheorem{prop}[thm]{Proposition}

\theoremstyle{definition}

\newtheorem{rem}[thm]{Remark}

\numberwithin{equation}{section}

\newcommand{\bee}{\begin{equation}}
	\newcommand{\ene}{\end{equation}}
\newcommand{\bees}{\begin{equation*}}
	\newcommand{\enes}{\end{equation*}}
\newcommand{\bes}{\begin{split}}
	\newcommand{\ens}{\end{split}}

\newcommand{\bet}{\begin{thm}}
	\newcommand{\ent}{\end{thm}}
\newcommand{\bel}{\begin{lem}}
	\newcommand{\enl}{\end{lem}}
\newcommand{\bec}{\begin{cor}}
	\newcommand{\enc}{\end{cor}}
\newcommand{\bep}{\begin{proof}}
	\newcommand{\enp}{\end{proof}}
\newcommand{\ber}{\begin{rem}}
	\newcommand{\enr}{\end{rem}}

\newcommand{\La}{\Lambda}

\newcommand{\al}{\alpha}

\newcommand{\Ga}{\Gamma}

\begin{document}
	\hoffset -4pc

\title
[Eigenvalue Bounds for Reduced Density Matrices]
{Eigenvalue Bounds for Multi-Particle Reduced Density Matrices of Coulombic Wavefunctions}
\author{Peter Hearnshaw}
\address{Centre for the Mathematics of Quantum Theory\\ University of Copenhagen\\
	Universitetsparken 5\\ DK-2100 Copenhagen \O\\ Denmark}
\email{ph@math.ku.dk}

\begin{abstract}	 
For bound states of atoms and molecules we consider the corresponding $K$-particle reduced density matrices, $\Ga^{(K)}$, for $1 \le K \le N-1$. Previously, eigenvalue bounds were obtained in the case of $K=1$ and $K=N-1$ by A.V. Sobolev. The purpose of the current work is to obtain bounds in the case of $2 \le K \le N-2$. For such $K$ we label the eigenvalues of the positive, trace class operators $\Ga^{(K)}$ by $\lambda_n(\Ga^{(K)})$ for $n=1,2,\dots$, and obtain the bounds $\lambda_n(\Ga^{(K)}) \le Cn^{-\al_K}$ for all $n$, where $\al_K = 1 + 7/(3L)$ and $L = \min\{K,N-K\}$.
\end{abstract}


\maketitle

\section{Introduction}

We consider an atom with $N$ electrons and a nucleus of charge $Z >0$ fixed at the origin. The non-relativistic Schr\oo dinger equation for this system is
\begin{align}
\label{eq:se}
(-\Delta + V)\psi = E\psi
\end{align}
for $\Delta$ the Laplacian on $\real^{3N}$ and $V$ the Coulomb potential given by
\begin{align}
\label{eq:v}
V(\bx) = -\sum_{j=1}^N \frac{Z}{|x_j|} + \sum_{1 \le j<k \le N}\frac{1}{|x_j-x_k|}
\end{align}
where $\bx = (x_1, \dots, x_N) \in \real^{3N}$ is the collection of $\real^3$ positions of the $N$ particles. We consider $\psi$ to be a normalized eigenfunction of this equation in $L^2(\real^{3N})$, and we impose no symmetry assumptions. All results can be extended to the case of molecules.

We now define the \textit{reduced density matrices} for such $\psi$. Firstly, for each permutation $\s$ in $S_N$, the symmetric group of degree $N$, we set $P_{\s} : \real^{3N} \to \real^{3N}$ given by
\begin{align*}
P_{\s}(\bx) = (x_{\s(1)}, \dots, x_{\s(N)}), \quad \bx \in \real^{3N}.
\end{align*}
Let $K=1,\dots, N-1$. Then, for a given $\s \in S_N$, we can define the function
\begin{align*}
\g_{\s}^{(K)}(\check\bx, \check\by) = \int_{\real^{3(N-K)}}\overline{\psi(P_{\s}(\check\bx,\hbx))} \psi(P_{\s}(\check\by, \hbx))\,d\hbx
\end{align*}
for $\check\bx,\, \check\by \in \real^{3K}$. We set $\g^{(K)} = \g_I^{(K)}$ where $I$ is the identity in $S_N$. Indeed, for regularity concerns the choice of $\s$ is unimportant so often $\s = I$ is taken. In fact, for $\psi$ totally symmetric or antisymmetric we get $\g_{\s}^{(K)} = \g^{(K)}$ for any $\s$.

This functions $\g_{\s}^{(K)}$ are well-defined and are in fact Lipschitz continuous due to $\psi$ being Lipschitz continuous as shown by Kato, \cite{kato57}. These functions are analytic outside of a singular set, \cite{hs22}, \cite{jec22}. There have also been many recent results concerning the regularity at the singular sets, in particular at the ``diagonal'', namely \cite{cio20}, \cite{cio22}, \cite{hs23}, \cite{he}, \cite{jec23}, \cite{jec24}.
  
The $K$-\textit{particle reduced density matrix} is then defined as the integral operator $\bold{\Ga}^{(K)} : L^2(\real^{3K}) \to L^2(\real^{3K})$ given by
\begin{align}
\label{eq:dm1}
\big(\bold{\Ga}^{(K)}u\big)(\check\bx) = \frac{1}{(N-K)!}\sum_{\s \in S_N}\big(\G^{(K)}_{\s}u\big)(\check\bx), \quad u \in L^2(\real^{3K}),
\end{align}
where
\begin{align}
\label{eq:dm2}
\big(\G^{(K)}_{\s}u\big)(\check\bx) = \int_{\real^{3K}}\g_{\s}^{(K)}(\check\bx, \check\by)u(\check\by)\,d\check\by.
\end{align}
for $\check\bx \in \real^{3K}$. Furthermore, we define $\G^{(K)} = \G^{(K)}_I$ for $I$ the identity in $S_N$. We note that the factor $1/(N-K)!$ appears due to overcounting, since for many $\s$ the integrals are identical by Fubini's theorem. If $\psi$ is totally symmetric or antisymmetric then formula (\ref{eq:dm1}) may be simplified by replacing the summation by the factor $N!$. The operators $T = \bold{\Ga}^{(K)}$ or $\G^{(K)}_{\s}$ are positive and trace class, and we label their eigenvalues as $\lambda_n(T)$, non-increasing counting multiplicity.

In our results, we will require that $\psi$ decays exponentially. This is really no restriction at all since it holds true whenever $E$ is a discrete eigenvalue, for a detailed discussion see \cite{simon}. More precisely, we consider $\psi$ such that there exists $C, \kappa>0$ such that
\begin{align}
\label{eq:exp_decay}
|\psi(\bx)| \le Ce^{-\kappa|\bx|'},\quad \text{where}\quad |\bx|' = |x_1| + \dots + |x_N|
\end{align}
for all $\bx \in \real^{3N}$. The use of $|\,\cdot\,|'$ rather than the standard Euclidean norm $|\,\cdot\,|$ in $\real^{3N}$ is simply for later convenience.
	
We are interested in the eigenvalue decay of the operators $\Ga^{(K)}$ for $1 \le K \le N-1$. The first significant result in this direction seems to be \cite{frie03}, where it was shown that $\Ga^{(1)}$ has infinite rank for $\psi$ a fermionic ground state. Using regularity estimates for $\psi$ given by S. Fournais and T.\O. S\o rensen in \cite{fs21}, bounds on the eigenvalue decay were then obtained by A.V. Sobolev \cite{sob22_1} to give the following theorem.
\begin{thm}\cite[Theorem 1.1]{sob22_1}
\label{thm:sobolev}
Let $\psi$ be any eigenfunction of (\ref{eq:se}) obeying the exponential decay condition (\ref{eq:exp_decay}). Then there exists $C$ such that, for any $n \ge 1$,
\begin{align}
\lambda_n(\Ga^{(1)}),\, \lambda_n(\Ga^{(N-1)}) \le Cn^{-8/3}.
\end{align}
\end{thm}
In the same manner as described in Remark \ref{rem:1}(1) below, the above bounds also immediately hold for $\Ga_{\s}^{(1)}, \Ga_{\s}^{(N-1)}$ for any $\s \in S_N$, and also for $\bold{\Ga}^{(1)}$ and $\bold{\Ga}^{(N-1)}$. The result for $\Ga^{(N-1)}$ was not specifically mentioned in \cite{sob22_1}, but is an immediate consequence of (\ref{eq:psi_lam}) below. The results of Theorem \ref{thm:sobolev} were extended in \cite{sob22_2} to give asymptotics for $\bold{\Ga}^{(1)}$ and $\bold{\Ga}^{(N-1)}$ with the same eigenvalue decay rate. See also \cite{sob22_3} for corresponding result for the kinetic energy density matrix.

The following is our main theorem.
\begin{thm}
\label{thm:main}
Let $\psi$ be any eigenfunction of (\ref{eq:se}) obeying the exponential decay condition (\ref{eq:exp_decay}). For each $2 \le K \le N-2$ there exists $C$ such that, for all $n \ge 1$,
\begin{align}
\label{eq:main}
\lambda_n(\Ga^{(K)}) \le Cn^{-\al_K},\quad  \al_K = 1 + \frac{7}{3L}
\end{align}
where $L = \min\{K,N-K \}$.
\end{thm}

\begin{rem}
\begin{enumerate}
\label{rem:1}
\item The same eigenvalue bounds hold for $\Ga_{\s}^{(K)}$ for any $\s \in S_N$, since the proof is not sensitive to such a modification. Furthermore, the same bounds hold for $\bold{\Ga}^{(K)}$. Indeed, the condition (\ref{eq:main}) is precisely that $\Ga^{(K)} \in S_{1/\al_K,\,\infty}$ and these are vector spaces, see Section \ref{chpt:sp}.
\item $\Ga^{(K)}$ is trace class for each $K$. Hence, by (\ref{eq:sp_incl0}) we immediately get the bound (\ref{eq:main}) with $\al_K = 1$. We therefore obtain an improvement to this for every $K$.
\item We do not believe the above values of $\al_K$ for $2 \le K \le N-2$ are optimal. However, we do not expect these can chosen smaller than the optimal value for $K=1$, namely $\al_1=8/3$. We provide an argument based on regularity of the kernels $\g^{(K)}$, see the introduction for relevant references. As pointed out in \cite[Remark 1.2(7)]{hs23} the value $\al_1=8/3$ is seen to arise from a fifth order cusp for $\g^{(1)}(x_1, y_1)$ at the diagonal $x_1 \approx y_1$. Consider, for example, $K=2$ and take some non-zero $z,z'$ with $z \ne z'$. Then, using the same arguments which demonstrated a fifth order cusp for $\g^{(1)}$, we would get that $\g^{(2)}(x_1, z, y_1, z')$ has a fifth order cusp around $x_1 \approx y_1$ for $x_1, y_1$ separated from $0, z$ and $z'$. Due to the presence of the same type of singularity, we therefore do not expect a stronger decay rate of eigenvalues.
\item Eigenvalue bounds (\ref{eq:main}) can tell you how well $\Ga^{(K)}$ can be approximated by projections of finite rank. This is of particular interest in the $K=2$ case, where one may obtain the eigenvalue of $\psi$ from knowledge of $\bold{\Ga}^{(2)}$ alone via an explicit formula, \cite[(3.1.34)]{ls10}.
\end{enumerate}
\end{rem}
The proof proceeds by factorizing $\Ga^{(K)}$ as the product $(\Psi^{(K)})^*\Psi^{(K)}$ for a Hilbert-Schmidt operator $\Psi^{(K)}$, defined later in (\ref{eq:psi_k}). It will then suffice to consider bounds to the singular values of $\Psi^{(K)}$. In particular we need inclusion into certain Schatten classes, $S_p$ or $S_{p,\infty}$, defined in Section \ref{chpt:sp}.

The kernel of $\Psi^{(K)}$ can be factorized appropriately using Jastrow factors, a common technique used in the study of the regularity of wavefunctions. One factor has improved smoothness compared to the original kernel as shown by the enhanced pointwise bounds given in Proposition \ref{prop:phi_bound}. This can be shown to locally lie in a Besov-Nikol'skii space of good smoothness. Such spaces are introduced in Section \ref{chpt:besov}. The singular values of the integral operator with this smoothened factor as its kernel can be estimated using the Birman-Solomyak bounds given in \cite{bs77}. Finally, to get the kernel of $\Psi^{(K)}$ we need to multiply with the remaining factors. Some can be considered ``weights'' to the Birman-Solomyak bounds, and these do not affect the decay of the singular values. Others will be considered ``multipliers'' on $S_p$ or $S_{p,\infty}$. These are functions with the property that they can multiply a kernel of an operator in $S_p$ or $S_{p,\infty}$, and the resulting operator will lie in the same space.
%
\\\\
\textbf{Notation}. For each $1 \le K \le N$ we can define
\begin{align}
\label{eq:sig_k}
\Sigma^{(K)} = \Big\{ \bx \in \real^{3K} : \prod_{1 \le j \le K}|x_j| \prod_{1 \le j<k\le K}|x_j-x_k| = 0 \Big\}.
\end{align}
For simplicity, we will denote $\Sigma = \Sigma^{(N)}$ which is the singular set of the Coulomb potential, $V$, given in (\ref{eq:v}). By elliptic regularity, it's straightforward to show that if $\psi$ is a weak solution to (\ref{eq:se}) then $\psi \in C^{\infty}(\real^{3N}\backslash\Sigma)$.

Given some $1 \le K \le N-1$, we will consider $\check\bx \in \real^{3K}$ and $\hbx \in \real^{3(N-K)}$ whose $\real^3$ components are labelled as follows,
\begin{align}
\label{not:x1}
\check\bx = (x_1, \dots, x_K), \quad \hbx = (x_{K+1}, \dots, x_N).
\end{align}
Therefore, $(\check\bx, \hbx) = \bx$ for $\bx = (x_1, \dots, x_N)$ an element of $\real^{3N}$.

Now, for any $j \in \{1, \dots, K\}$ and $l \in \{K+1, \dots, N\}$ we remove a variable from either $\check\bx$ or $\hbx$ using the following subscript notation,
\begin{align}
\label{not:x2}
\check\bx_j = (x_1, \dots, x_{j-1}, x_{j+1}, \dots, x_K), \quad \hbx_l = (x_{K+1}, \dots, x_{l-1}, x_{l+1}, \dots, x_N)
\end{align}
with obvious modifications for $j=1$ or $K$, and for $l=K+1$ or $N$.
We also write
\begin{align}
\label{not:x3}
(x_j, \check\bx_j) = \check\bx, \quad (x_l, \hbx_l) = \hbx.
\end{align}
In other words, we understand $x$ in $(x,\check\bx_j)$ to be placed as the $j$-th component. Similarly for $(x,\hbx_l)$.
\subsection{Proof of Theorem \ref{thm:main}.} We begin by writing the reduced density matrices as products of operators. In the same manner as \cite{sob22_1}, we define the operator $\Psi_K : L^2(\real^{3K}) \to L^2(\real^{3(N-K)})$ by
\begin{align}
\label{eq:psi_k}
(\Psi^{(K)} u)(\hbx) = \int_{\real^{3K}} \psi(\check\bx, \hbx) u(\check\bx)\,d\check\bx, \quad u \in L^2(\real^{3K})
\end{align}
for $\hbx \in \real^{3(N-K)}$. Since $\psi \in L^2$, this is readily seen to be a well-defined Hilbert-Schmidt operator.

By straightforward calculations,
\begin{align*}
\big(\Psi^{(K)}\big)^*\Psi^{(K)} = \Ga^{(K)},\qquad \Psi^{(K)}\big(\Psi^{(K)}\big)^* = \overline{\Ga_{\s_K}^{(N-K)}}
\end{align*}
where $\s_K$ is such that $P_{\s_K}(\hbx, \check\bx) = (\check\bx, \hbx)$ for any $\check\bx \in \real^{3K}$, $\hbx \in \real^{3(N-K)}$, and the complex conjugation of the operator is that of the kernel. By the definition of singular values we then get
\begin{align}
\label{eq:psi_lam}
s_n\big(\Psi^{(K)}\big)^2 = \lambda_n\big(\Ga^{(K)}\big) = \lambda_n\big(\Ga_{\s_K}^{(N-K)}\big),\quad n \in \naturals.
\end{align}
Here we also used that $s_n(T^*) = s_n(T)$ for compact $T$, see Section \ref{chpt:sp}, and that $\lambda_n(T) = \lambda_n(\overline{T})$ for a positive compact integral operator.

It therefore sufficies to study the singular values of $\Psi^{(K)}$. The appropriate inclusion into $S_{q,\infty}$ spaces, as defined in Section \ref{chpt:sp}, is given by the following proposition.
\begin{prop}
\label{prop:psi_k}
Let $2 \le K \le N-2$. Then $\Psi^{(K)} \in S_{q,\infty}$ for $1/q = 1/2 + 7/(6K)$.
\end{prop}

The proof will be given in Section \ref{chpt:psi_k}, but we now use this to prove our main theorem.
\begin{proof}[Proof of Theorem \ref{thm:main}]
By (\ref{eq:psi_lam}), Proposition \ref{prop:psi_k}, and the definition of $S_{q,\infty}$ spaces, see Section \ref{chpt:sp}, we get some $C$ such that
\begin{align*}
\lambda_n\big(\Ga^{(K)}\big),\, \lambda_n\big(\Ga_{\s_K}^{(N-K)}\big) \le Cn^{-(1+7/(3K))}
\end{align*}
for all $n \ge 1$. Since we could have taken any permuation of the variables of $\psi$ in defining $\Psi^{(K)}$, it follows that we can freely replace $\Ga_{\s_K}^{(N-K)}$ by $\Ga^{(N-K)}$ in the above bound. This completes the proof.
\end{proof}

\section{Factorization of $\psi$}
It is known that multiplication by a so-called Jastrow factor can improve the smoothness of solutions $\psi$ to (\ref{eq:se}), see, for example, \cite{hos01}, \cite{fhos05}, \cite{fs21}. We follow these authors to define, for $\bx \in \real^{3N}$,
\begin{align}
\label{eq:f_def}
F(\bx) = -\frac{Z}{2}\sum_{1 \le j \le N}\tau(x_j) + \frac{1}{4}\sum_{1\le j<k \le N}\tau(x_j-x_k),
\end{align}
where
\begin{align}
\label{eq:tau}
\tau(x) = |x| - (1+|x|^2)^{1/2}, \quad x \in \real^3.
\end{align}
Using this, we set
\begin{align}
\label{eq:phi_def}
\phi = e^{-F}\psi.
\end{align}

The function $\phi$ lies in $C^{1,\al}(\real^{3N})$ for every $\al \in [0,1)$. In \cite{fhos05} it was shown that an additional term may be added to $F$ in (\ref{eq:f_def}) to get $C^{1,1}$ regularity, which is equivalent to the first two derivatives being locally bounded. We're interested in pointwise bounds to the derivatives of $\phi$. The first result of this kind was from \cite{fs21}, and this method was adapted in \cite{hs22} and \cite{he}. We require a slightly different result but rather than go through the method in full, since it is fairly long, we instead obtain the following bounds by using elements of \cite{hs22} and applying \cite{fhos05}. This is presented in Appendix \ref{chpt:app}.
\begin{prop}
\label{prop:phi_bound}
Let $\kappa$ be such that (\ref{eq:exp_decay}) holds and suppose $\phi$ is defined as in (\ref{eq:phi_def}). Let $j \in \{1,\dots,N\}$. Then for each $m \in \naturals_0^3$ there exists $C_{m}$, which may depend on $\psi$, such that
\begin{align}
\label{eq:phi_bound}
|\partial^{m}_{x_j}\phi(\bx)| \le C_{m}e^{-\kappa |\bx|'} \Big(1+|x_j|^{\min\{2-|m|,\,0\}} + \sum_{\substack{1 \le k \le N \\ k \ne j}}|x_j-x_k|^{\min\{2-|m|,\,0\}}\Big)
\end{align}
for all $\bx \in \real^{3N}\backslash\Sigma$.
\end{prop}

It will be convenient for us to extract the exponential decay. In particular, define
\begin{align}
\label{eq:mu_def}
\mu(\bx) = \phi(\bx) \prod_{l=1}^N e^{\kappa (1+|x_l|^2)^{1/2}}.
\end{align}
Using $m=0$ in (\ref{eq:phi_bound}) we find that
\begin{align*}
|\mu(\bx)| \le 3C_0 \prod_{l=1}^N e^{-\kappa \tau(x_l)}
\end{align*}
which is bounded in $\real^{3N}$ since $\tau \in L^{\infty}(\real^3)$. Let $j \in \{1, \dots, N\}$. Then we can continue, by straightforward calculations, to obtain, for any $m \in \naturals^3_0$, some $C'_m$ such that
\begin{align}
\label{eq:mu}
|\partial^m_{x_j}\mu(\bx)| \le C'_{m}\Big(1+|x_j|^{\min\{2-|m|,\,0\}} + \sum_{\substack{1 \le k \le N \\ k \ne j}}|x_j-x_k|^{\min\{2-|m|,\,0\}}\Big)
\end{align}
for all $\bx \in \real^{3N}\backslash\Sigma$.

We now proceed to find an appropriate factorisation of $\psi$. For every $1 \le K \le N-1$ we will use the same $\mu$ which will be the part of $\psi$ with the greatest smoothness. This is what will determine the singular value estimates for $\Psi^{(K)}$, defined in (\ref{eq:psi_k}), at least for $K \ge 2$. The remaining factors are partitioned based on $K$ in the following manner.

Let $1 \le K \le N-1$. As in (\ref{not:x1}), we write $\check\bx = (x_1, \dots, x_K)$, $\hbx = (x_{K+1}, \dots, x_N)$. Then we set
\begin{align}
\label{eq:a}
\mathcal{A}(\check\bx) &= \exp\Big(-\frac{Z}{2}\sum_{l=1}^K \tau(x_l) + \frac{1}{4}\sum_{l=1}^{K-1}\sum_{m=l+1}^K \tau(x_l-x_m)\Big) \prod_{l=1}^K e^{-\kappa (1+|x_l|^2)^{1/2}},\\
\label{eq:b}
\mathcal{B}(\hbx) &= \exp\Big(-\frac{Z}{2}\sum_{l=K+1}^N \tau(x_l) + \frac{1}{4}\sum_{l = K+1}^{N-1} \sum_{m=l+1}^N \tau(x_l-x_m)\Big)\prod_{l=K+1}^N e^{-\kappa (1+|x_l|^2)^{1/2}},\\
\label{eq:c}
\mathcal{C}(\check\bx, \hbx) &= \exp\Big(\frac{1}{4}\sum_{l=1}^K \sum_{m=K+1}^N\tau(x_l-x_m)\Big).
\end{align}
By (\ref{eq:phi_def}) and (\ref{eq:mu_def}) we then obtain
\begin{align}
\label{eq:psi_decomp}
\psi(\check\bx, \hbx) = \mathcal{A}(\check\bx)\mathcal{B}(\hbx) \mathcal{C}(\check\bx, \hbx) \mu(\check\bx, \hbx).
\end{align}
With this decomposition of the kernel of $\Psi^{(K)}$ we will apply the Birman-Solomyak bounds of singular values, Proposition \ref{prop:bs}, applied to the kernel $\mu$ with weights $\mathcal{A}$ and $\mathcal{B}$. The function $\mathcal{C}$ will be seen to act as a ``multiplier'' in the sense of Section \ref{chpt:mult}, and will not affect the obtained singular value bounds.

\section{Schatten classes and integral operators}
\label{chpt:sp}
We now introduce singular values and Schatten classes of compact operators, for more details see \cite[Chapter 11]{bs87}. Let $\mathcal{H}_1, \mathcal{H}_2$ be separable Hilbert spaces and $T : \mathcal{H}_1 \to \mathcal{H}_2$ a compact operator. The singular values of $T$ are $s_n(T) = \lambda_n(T^*T)^{1/2}$, $n \ge 1$, where $\lambda_n(T^*T)$ are the non-zero eigenvalues of $T^*T$ arranged non-increasing and counting multiplicity. By the Schmidt representation of a compact operator we get $s_n(T^*) = s_n(T)$ for all $n$. For $0<q<\infty$, the $q$\textit{-th Schatten class} is the space $S_q = S_q(\mathcal{H}_1, \mathcal{H}_2)$ defined as all compact $T$ such that the quantity
\begin{align*}
\norm{T}_q = \Big(\sum_n s_n(T)^q\Big)^{1/q}
\end{align*}
is finite. This functional defines a norm on $S_q$ for $q \ge 1$ and in this case the spaces are Banach. For $0<q<\infty$ we can also define the space $S_{q,\infty} = S_{q,\infty}(\mathcal{H}_1, \mathcal{H}_2)$ as all compact $T$ for which $s_n(T) = O(n^{-1/q})$ as $n \to \infty$. On such spaces we can define the functional 
\begin{align}
\label{eq:spinf}
\norm{T}_{q,\infty} = \sup_n n^{1/q} s_n(T).
\end{align}

We now state some inclusions. Suppose $T \in S_p$, then $T \in S_{p,\infty}$ with
\begin{align}
\label{eq:sp_incl0}
\norm{T}_{p,\infty} &\le \norm{T}_p,
\end{align}
and for any $q>p$ we have $T \in S_q$ with
\begin{align}
\label{eq:sp_incl}
\norm{T}_q \le \norm{T}_p.
\end{align}
Furthermore, if $T \in S_{p,\infty}$ then for any $q>p$ we have $T \in S_q$ and
\begin{align}
\label{eq:sp_incl2}
\norm{T}_{q} \le C(q,p)\norm{T}_{p, \infty}.
\end{align}

We use the following triangle-type inequality which will be sufficient for our purposes. For more details see \cite[Chapter 11]{bs87}, \cite[Chapter 1]{bs77} and \cite{alek02}.
\begin{prop}\cite[Lemma 1.1]{bs77}
\label{prop:triangle}
Let $q \in (0,2)$ and let $T_j \in S_{q,\infty}$, $j=1,\dots$, be a family of operators satisfying either $T_j^* T_k = 0$ for all $j \ne k$, or $T_j T^*_k = 0$ for all $j \ne k$.
Then $T = \sum_j T_j \in S_{q,\infty}$ and
\begin{align*}
\norm{T}^q_{q,\infty} \le \frac{2}{2-q}\sum_j \norm{T_j}^q_{q,\infty},
\end{align*}
under the assumption the right-hand side is finite.
\end{prop}
\subsection{Integral operators.}
Let $X \subset \real^d$, $Y \subset \real^m$. Let $T : Y \times X \to \complex$, $b : Y \to \complex$ and $a : X \to \complex$. Then, formally we define $\intt(bTa) : L^2(X) \to L^2(Y)$ as
\begin{align}
\label{eq:io_def}
\intt(bTa)u(y) = \int_X b(y)T(y,x)a(x)u(x)\,dx.
\end{align}
If $a=b\equiv 1$ we will simply write $\intt(T)$. It is well known that if $T \in L^2(Y \times X)$ then $\intt(T) \in S_2$. By imposing a smoothness condition on $T$, one may then obtain inclusion in the more refined spaces $S_{q,\infty}$ for $q<2$.  This is seen in the following proposition, due to M.S. Birman and M.Z. Solomyak. The statement was obtained from \cite[Corollaries 4.1 and 4.5, Theorems 4.4 and 4.7]{bs77}. It involves the Besov-Nikol'skii spaces, $N_2^s$, introduced in the next section.
\begin{prop}
\label{prop:bs}
Let $X \subset \real^d$ be a bounded Lipschitz domain. Assume that the kernel $T(y,x)$, $y \in \real^m$, $x \in X$, is such that $T(y,\,\cdot\,) \in N^s_2(X)$ with some $s>0$, for a.e. $y \in \real^m$. Assume that $b \in L^2_{loc}(\real^m)$ and that $a \in L^r(X)$, where
\begin{align*}
\begin{cases}
r=2, &\text{if } 2s>d\\
r>d/s, &\text{if } 2s\le d.
\end{cases}
\end{align*}
Then $\intt(bTa) : L^2(X) \to L^2(\real^m)$ lies in $S_{q,\infty}$ where $1/q = 1/2 + s/d$ and
\begin{align*}
\norm{\intt(bTa)}_{q,\infty} \le C\Big(\int_{\real^m} \norm{T(y,\,\cdot\,)}^2_{N^s_2(X)}|b(y)|^2 \,dy \Big)^{1/2} \norm{a}_{L^r(X)}
\end{align*}
under the assumption the right-hand side is finite.
\end{prop}

\section{Besov-Nikol'skii spaces}
\label{chpt:besov}
Take a function $u \in L^q(\real^d)$, $1 \le q \le \infty$, and let $h \in \real^d$. Then define the \textit{finite difference of order 1} as
\begin{align*}
\Delta_h^{(1)}u(x) = u(x+h)-u(x), \qquad x \in \real^d.
\end{align*}
For $l=1,2,\dots$ we then define the \textit{finite difference of order }$l$, denoted by $\Delta_h^{(l)}$, as $l$ successive applications of $\Delta_h^{(1)}$. This produces the formula
\begin{align*}
\Delta_h^{(l)}u(x) = \sum_{j=0}^l (-1)^{j+l} {l \choose j} u(x + jh).
\end{align*}
Additionally, we take $\Delta_h^{(0)}$ as the identity. 

Since $\norm{\Delta^{(1)}_h u}_q \le 2\norm{u}_q$, it is immediate that
\begin{align}
\label{eq:reduce_fd}
\norm{\Delta^{(k)}_h u}_q = \norm{\Delta^{(k-l)}_h \Delta^{(l)}_h u}_q \le 2^{k-l}\norm{\Delta^{(l)}_h u}_q
\end{align}
for integers $k>l$.

We now introduce the \textit{Besov-Nikol'skii} spaces $N^s_q$, which exist within the spectrum of Besov spaces, and in this setting are commonly labelled $B^s_{q,\infty}$. Further details on these spaces may be found in \cite{bs77}, \cite{adams}, \cite{herz68}, \cite{bergh}. Let $s > 0$. A function $u$ is said to lie in $N^s_q(\real^d)$, $1 \le q \le \infty$, if $u \in L^q(\real^d)$ and there is some integer $l > s$ such that
\begin{align*}
[u]_{s,l} := \sup_{h \in \real^d,\, h \ne 0}|h|^{-s}\norm{\Delta^{(l)}_h u}_q < \infty.
\end{align*}
In fact, by the following lemma this shows $[u]_{s,l}$ will be finite for all integers $l>s$. A short proof can be found in \cite[Lemma 1.1]{herz68}.
\begin{lem}
\label{lem:ns}
Let $u \in L^q(\real^d)$, $1 \le q \le \infty$, and suppose $[u]_{s,l} < \infty$ for some integer $l>s$. Then for any other integer $m>s$ we have $C$, depending only on $s,l$ and $m$, such that
\begin{align}
\label{eq:ns_equiv}
[u]_{s,m} \le C[u]_{s,l}.
\end{align}
\end{lem}
The spaces $N^s_q(\real^d)$ admit the norm
\begin{align}
\norm{u}_{N^s_q} = \norm{u}_{L^q} + [u]_{s, l}
\end{align}
for any choice of integer $l > s$, and these produce equivalent norms by the preceeding lemma. It's often convenient to take $l = [s]+1$.

For an open set $\Omega \subset \real^d$, we define the space $N^s_q(\Omega)$ as all functions $u \in L^q(\Omega)$ such that the following quantity is finite
\begin{align}
\label{eq:besov_rest}
\norm{u}_{N^s_q(\Omega)} = \inf \norm{\tilde u}_{N^s_q(\real^d)}
\end{align}
where the infimum is over functions $\tilde u \in N^s_q(\real^d)$ such that $u = \tilde u$ a.e. on $\Omega$. Other characterizations of such spaces on domains exist, but this will suffice for our purposes.

The following lemma is \cite[Lemma 3.1]{hs23}. This result gives a condition for a function to lie in $N^s_q$-spaces which will be convenient for our purposes.
\begin{lem}
\label{lem:du}
Let $d \ge 2$, $1 \le q \le d$ and let $X = \{a_j\}_{j=1}^M \subset \real^d$ be some collection of $M$ points. For a given $\al > -d/q$, we set $s = \al + d/q$. Suppose for some $A>0$ the function $u \in C^{\infty}(\real^d\backslash X)$ obeys
\begin{align*}
|\nabla^k u(x)| \le A \Big(1 + \sum_{j=1}^M |x-a_j|^{\min\{\al-k,\,0\}}\Big), \quad k =0,1,\dots, [s]+1
\end{align*}
for all $x \in \real^d\backslash X$. Then for any $z \in \real^d$ and $R>0$ we have $u \in N^s_q(B(z,R))$ and $\norm{u}_{N^s_q(B(z,R))} \le CA$ for some $C$ possibly depending on $M$ and $R$, but not $z$ or $X$ (for fixed $M$). 
\end{lem}

The following lemma gives a criteria for inclusion into $N^s_q$-spaces for functions whose restrictions to lower dimensional spaces have appropriate $N^s_q$-smoothness. It is elementary, but we have not managed to find it in the literature. It will play a vital part in the proof of our main result, Theorem \ref{thm:main}, since for particles in $\real^3$, Lemma \ref{lem:du} alone would only be able to prove $N^s_q$ inclusion on subsets of $\real^3$ and not $\real^{3K}$ in general, which we will require.
\begin{lem}
\label{lem:mb_besov}
Let $d_1, d_2 \ge 1$ and $d=d_1+d_2$. Suppose $u(\,\cdot\,,x_2) \in N_q^s(\real^{d_1})$ and $u(x_1, \,\cdot\,) \in N_q^s(\real^{d_2})$ for a.e. $x_1 \in \real^{d_1}$, $x_2 \in \real^{d_2}$, and that there is some $A>0$ such that
\begin{align}
\label{eq:ns_partial}
\Big(\int_{\real^{d_1}}\norm{u(x_1, \,\cdot\,)}^q_{N^s_q(\real^{d_2})}\,dx_1\Big)^{1/q} +\Big(\int_{\real^{d_2}}\norm{u(\,\cdot\,, x_2)}^q_{N^s_q(\real^{d_1})}\,dx_2\Big)^{1/q} \le A,
\end{align}
with the appropriate modification for $q=\infty$. Then $u \in N^s_q(\real^d)$ and $\norm{u}_{N^s_q(\real^d)} \le CA$. Here, $C$ is independent of  $A$.
\end{lem}
\begin{rem}
It is interesting to consider the opposite problem, that is, given a general $u \in N^s_q(\real^d)$, is it true that $u(\,\cdot\,,x_2) \in N^s_q(\real^{d_1})$ for a.e. $x_2 \in \real^{d_2}$? This is, in fact, false in general, see \cite{brass18} and \cite{mrs20}.
\end{rem}

\begin{proof}
By the definition of the $N^s_q$-norm, we can use either term on the left-hand side of (\ref{eq:ns_partial}) to immediately show that $\norm{u}_{L^q(\real^d)} \le A$. It remains to consider the quantity $[u]_{s,l}$ for an appropriate $l$.

Take any $\bh = (h_1,h_2) \in \real^{d_1} \times \real^{d_2}$ and set $\bh_1 = (h_1,0)$ and $\bh_2 = (0,h_2)$. Then a simple calculation shows
\begin{align*}
\Delta_{\bh}^{(1)}u(x_1,x_2) = \Delta_{\bh_1}^{(1)}u(x_1,x_2) + \Delta_{\bh_2}^{(1)}u(x_1+h_1, x_2),
\end{align*}
and hence by iteration,
\begin{align*}
\Delta_{\bh}^{(l)}u(x_1,x_2) = \sum_{k=0}^l {l \choose k} \Delta_{\bh_1}^{(l-k)} \Delta_{\bh_2}^{(k)}u(x_1+kh_1, x_2).
\end{align*}

Let $r = [s]+1$. Using the above formula for $l=2r$ we get
\begin{align*}
\norm{\Delta_{\bh}^{(2r)}u}_{L^q(\real^{d})} \le C\sum_{k=0}^{2r} \norm{\Delta_{\bh_1}^{(2r-k)} \Delta_{\bh_2}^{(k)}u}_{L^q(\real^d)}.
\end{align*}
We choose finite differences of order $2r$ so that in each term of the above sum, at least one of $\Delta_{\bh_1}$ and $\Delta_{\bh_2}$ is of order at least $r$.  Without loss, we therefore consider terms where $k \ge r$. Terms with $k<r$ are treated in a similar way, but the roles of $\bh_1$ and $\bh_2$ are interchanged. Using (\ref{eq:reduce_fd}), we remove the finite difference involving $\bh_1$ and reduce the order of $\Delta_{\bh_2}$ by $k-r$,
\begin{align*}
\norm{\Delta_{\bh_1}^{(2r-k)} \Delta_{\bh_2}^{(k)}u}_{L^q(\real^d)} \le 2^{r} \norm{\Delta_{\bh_2}^{(r)}u}_{L^q(\real^d)}.
\end{align*}
Suppose $q<\infty$. Then, for $u_{x_1} = u(x_1,\,\cdot\,)$ defined on $\real^{d_2}$ we have
\begin{align*}
|\bh|^{-qs}\norm{\Delta_{\bh_2}^{(r)}u}^q_{L^q(\real^d)} &\le |h_2|^{-qs} \int \norm{\Delta_{h_2}^{(r)}u_{x_1}}^q_{L^q(\real^{d_2})}\, dx_1\\
&\le  \int [u_{x_1}]^q_{s,r}\, dx_1 \le A^q.
\end{align*}
by (\ref{eq:ns_partial}). Similarly, for $q=\infty$ we get
\begin{align*}
|\bh|^{-s}\norm{\Delta_{\bh_2}^{(r)}u}_{L^{\infty}(\real^d)} \le A.
\end{align*}
Therefore, for some $C$,
\begin{align*}
\sup_{\bh \in \real^d,\, \bh \ne 0}|\bh|^{-s}\norm{\Delta_{\bh}^{(2r)}u}_{L^q(\real^{d})} \le CA,
\end{align*}
which completes the proof in view of Lemma \ref{lem:ns}.
\end{proof}

We now apply this result to the type of functions we're interested in. Suppose that $u \in C^{\infty}(\real^{3N}\backslash\Sigma)$. Given $\al > -3/2$, let $s = \al + 3/2$. In the following, $|\,\cdot\,|$ is understood in $\real^3$. Let $A>0$ be such that for each $j=1,\dots,N$ and each $m \in \naturals_0^3$, $|m| \le [s]+1$,
\begin{align}
\label{eq:mb_pointwise}
|\partial^{m}_{x_j}u(\bx)| \le A \Big(1+|x_j|^{\min\{\al-|m|,\,0\}} + \sum_{\substack{1 \le k \le N \\ k \ne j}}|x_j-x_k|^{\min\{\al-|m|,\,0\}}\Big)
\end{align}
for all $\bx \in \real^{3N}\backslash\Sigma$.

\begin{lem}
\label{lem:pointwise}
Let $\al >-3/2$ and $s = \al + 3/2$. Suppose $u \in C^{\infty}(\real^{3N}\backslash\Sigma)$ obeys (\ref{eq:mb_pointwise}) for some $A>0$. Let $1 \le K \le N$, and $X_1, \dots, X_K \subset \real^3$ be open and bounded sets. Then there exists $C$ such that
\begin{align*}
\norm{u(\,\cdot\,, \hbx)}_{N^s_2(X_1 \times \dots \times X_K)} \le CA, \quad \text{for all}\quad \hbx \in \real^{3(N-K)}\backslash\Sigma^{(N-K)}
\end{align*}
The constant $C$ is dependent on the sets $X_l$ only via $\diam(X_l)$, $1 \le l \le K$.
\end{lem}
For $K=N$, the function $u(\,\cdot\,, \hbx)$ should be understood simply as $u$. We use the standard definition of $\diam(X) = \sup\{|x-y| : x,y \in X\}$, for a set $X$.

\begin{proof}
We fix some $\hbx \in \real^{3(N-K)}\backslash\Sigma^{(N-K)}$ throughout, whose components we label as $\hbx = (x_{K+1}, \dots, x_N)$. Using this, we define the function $u_{\hbx}$ on $\real^{3K}$ by $u_{\hbx}(\check\bx) = u(\check\bx, \hbx)$ for all $\check\bx = (x_1, \dots, x_K) \in \real^{3K}$.

We begin by localizing the function $u_{\hbx}$ on $X_1 \times \dots \times X_K$. Firstly, let $\chi \in C_c^{\infty}(\real)$ obey $0 \le \chi \le 1$ and be such that $\chi(x) = 1$ for $|x| \le 1$ and $\chi(x) = 0$ for $|x| \ge 2$. For each $j=1,\dots,K$ we define $\chi_j \in C_c^{\infty}(\real^3)$ as follows. Take some $z_j \in X_j$ and set $d_j =\textnormal{diam}(X_j)$, whereby $X_j \subset B(z_j, d_j)$. We then define $\chi_j(x) = \chi(|x - z_j|/d_j)$ for $x \in \real^3$. This function obeys $\chi_j = 1$ on $B(z_j, d_j)$ and $\chi_j = 0$ on $\real^3\backslash B(z_j, 2d_j)$.

Now, let
\begin{align}
\label{eq:v_def}
v(\check\bx) = \chi_1(x_1) \dots \chi_K(x_K) u_{\hbx}(\check\bx)
\end{align} 
for $\check\bx \in \real^{3K}$. Then, by (\ref{eq:mb_pointwise}) we have some constant $C$, independent of $A$ and the choice of $\hbx$, such that for all $j=1,\dots,K$ and $|m| \le [s]+1$,
\begin{align}
\label{eq:mb_pointwise2}
|\partial^{m}_{x_j} v(\check\bx)| \le CA \Big(1+|x_j|^{\min\{\al-|m|,\,0\}} + \sum_{\substack{1 \le k \le N \\ k \ne j}}|x_j-x_k|^{\min\{\al-|m|,\,0\}}\Big)
\end{align}
for $\check\bx$ such that $(\check\bx, \hbx) \in \real^{3N}\backslash\Sigma$.

We now take some $1 \le j \le K$ and fix some $\check\bx_j \in \real^{3(K-1)}$ such that $(\check\bx_j, \hbx) \in \real^{3(N-1)}\backslash\Sigma^{(N-1)}$. Then we can define $v_{\check\bx_j}$ by $v_{\check\bx_j}(x) = v(x, \check\bx_j)$ for $x \in \real^3$, where we used the notation (\ref{not:x3}). The inequality (\ref{eq:mb_pointwise2}) becomes, with the same $C$,
\begin{align}
|\partial_x^{m} v_{\check\bx_j}(x)| \le CA \Big(1+|x|^{\min\{\al-|m|,\,0\}} + \sum_{\substack{1 \le k \le N \\ k \ne j}}|x-x_k|^{\min\{\al-|m|,\,0\}}\Big)
\end{align}
for all $x \in \real^3\backslash\{0,x_1, \dots, x_{j-1}, x_{j+1}, \dots, x_N \}$, with obvious modification if $j=1$. Hence, by Lemma \ref{lem:du} with $q=2$ and, for example, $z = z_j$ and $R= 2d_j$, we obtain $v(\,\cdot\,, \check\bx_j) \in N^s_2(\real^3)$ and $\norm{v(\,\cdot\,, \check\bx_j)}_{N^s_2(\real^3)} \le CA$ for some new $C$. Here we used that $v(\,\cdot\,, \check\bx_j)$ has support in $B(z_j, 2d_j)$.

Since this holds for each $j=1,\dots,K$ we can apply Lemma \ref{lem:mb_besov}, successively $K-1$ times, to get $v \in N^s_2(\real^{3K})$ and $\norm{v}_{N^s_2(\real^{3K})} \le CA$, for some new constant $C$. The required result then follows
\begin{align*}
\norm{u(\,\cdot\,,\hbx)}_{N^s_2(X_1 \times \dots \times X_K)} \le \norm{v}_{N^s_2(\real^{3K})},
\end{align*}
which itself comes from (\ref{eq:v_def}) and (\ref{eq:besov_rest}). Finally, we note that the constants have been independent of the choice of $\hbx$.

\end{proof}

\section{Multipliers}
\label{chpt:mult}
Let $X \subset \real^d$, $Y \subset \real^m$. In this section we will be concerned with the spaces $S_q = S_q(L^2(X), L^2(Y))$ and $S_{q,\infty} = S_{q,\infty}(L^2(X), L^2(Y))$ for $0< q \le 2$. Since $S_q \subset S_2$, each operator in the former has an integral operator representation in $L^2(Y \times X)$. This is clearly also true for $S_{q,\infty}$ when $q<2$. We will use the same symbol $T$ for both the operator and its kernel, and we will use the notation $\intt(bTa)$ and $\intt(T)$ introduced in (\ref{eq:io_def}).

Let $\La : Y \times X \to \complex$ be a function. We call $\La$ a ``multiplier'' on $S_q$ if $\intt(\La T)$ $\big(= \intt(\La(y,x) T(y,x))\big) \in S_q$ for each $T \in S_q$ and, in addition,
\begin{align*}
M_{q}(\La) := \sup_{0 \ne T \in S_q} \frac{\norm{\intt(\La T)}_q}{\norm{T}_q} < \infty.
\end{align*}
We define multipliers on $S_{q,\infty}$ in the same way, but replacing $S_q$ by $S_{q,\infty}$ and using the functional $\norm{\,\cdot\,}_{q,\infty}$. Due to the characterization of $S_2$ as $L^2(Y \times X)$, it can be shown that the set of multipliers on $S_2$ is exactly $L^{\infty}(Y \times X)$.

The following proposition is from \cite[Corollary 8.2]{bs77}. By the preceeding remark, one consequence of this result is that any multiplier on $S_q$ for $q<2$ will lie in $L^{\infty}(Y \times X)$.
\begin{prop}
\label{prop:bs1}
If $\La$ is a multiplier on $S_q$ for $q < 2$, then it is also a multiplier on $S_p$ and $S_{p, \infty}$ for all $p \in(q,2]$. Furthermore,
$M_p(\La) + M_{p, \infty}(\La) \le CM_{q}(\La)$.
\end{prop}

The following result is from \cite[Theorems 8.1 and 8.2]{bs77}.
\begin{prop}
\label{prop:bs2}
Let $0<q \le 1$. Then $\La \in L^{\infty}(Y \times X)$ is a multiplier on $S_q$ if and only if $\intt(b \La a) \in S_q$ for all $a \in L^2(X)$ and $b \in L^2(Y)$, and
\begin{align*}
\sup_{\norm{a}_{L^2(X)} = \norm{b}_{L^2(Y)} = 1} \norm{\intt(b \La a)}_{q} < \infty
\end{align*}
In which case, this quantity concides with $M_q(\La)$.
\end{prop}

The next lemma uses the previous two propositions and concerns multipliers which depend on fewer variables than the underlying space.
\begin{lem}
\label{lem:multiplier}
Let $X_1 \subset \real^{d_1}$ be a bounded Lipschitz domain and $X_2 \subset \real^{d_2}$ be arbitrary. Let $\om \in L^{\infty}(\real^m \times X_1)$ be such that for some $2s>d_1$ we have $\om(y,\,\cdot\,) \in N^s_2(X_1)$ for a.e. $y \in \real^m$ and $y \mapsto \norm{\om(y, \,\cdot\,)}_{N^s_2(X_1)} \in L^{\infty}(\real^m)$. Then $\Om(y, x_1, x_2) := \om(y, x_1)$ is a multiplier in $S_q(L^2(X_1 \times X_2),\, L^2(\real^m))$ and $S_{q,\infty}(L^2(X_1 \times X_2),\, L^2(\real^m))$ for all $q \in (q_0, 2)$, where $1/q_0 = 1/2 + s/d_1$. Furthermore, there exists $C$ such that
\begin{align}
\label{eq:mult_ineq}
M_{q}(\Om) + M_{q,\infty}(\Om) \le C\esssup_{y \in \real^m}\norm{\om(y, \,\cdot\,)}_{N^s_2(X_1)}.
\end{align}
The constant $C$ is independent of $\om$ and $X_2$ but dependent on $q, s, d_1$ and $X_1$, albeit independent of translations thereof.
\end{lem}

The utility of this lemma can be demonstrated in an example. We use the same quantities as in the lemma. Suppose $X_2$ is bounded. It is immediate from the definition of the Besov-Nikol'skii spaces that for a.e. $y$ we have $\Om(y,\,\cdot\,) \in N^s_2(X_1 \times X_2)$ for the same $s$. Now, if we used Propostion \ref{prop:bs} directly on $\intt(b\Om a) : L^2(X_1 \times X_2) \to L^2(\real^m)$, for some appropriate $a$ and $b$, we'd get inclusion in $S_{q,\infty}$ for $1/q = 1/2 + s(d_1+d_2)^{-1}$. But $\Om$ doesn't depend on $x_2$, so we really should expect the value $1/q = 1/2 + s/d_1$, even if we still consider this as an operator on $L^2(X_1 \times X_2)$. This is what is shown here, at least in the case of multipliers.

\begin{proof}
By Proposition \ref{prop:bs1}, it suffices to prove $\Om$ is a multiplier on $S_q$ for each $q>q_0$ sufficiently small. Therefore, take some $q \in (q_0, 1)$. By Proposition \ref{prop:bs2} we must therefore show there exists $C$, independent of $\om$, such that
\begin{align}
\label{eq:int_ineq}
\norm{\intt(b \Om a)}_q \le C\esssup_{y \in \real^m}\norm{\om(y, \,\cdot\,)}_{N^s_2(X_1)}
\end{align}
for all $a \in L^2(X_1 \times X_2)$, $b \in L^2(\real^m)$ with $\norm{a}_{L^2(X_1 \times X_2)} = \norm{b}_{L^2(\real^m)} = 1$.

Let $a$, $b$ be as above. We claim that
\begin{align}
\label{eq:int_claim}
\norm{\intt(b \Om a)}_{S_q(L^2(X_1 \times X_2),\, L^2(\real^m))} = \norm{\intt(b \om \tilde{a})}_{S_q(L^2(X_1),\, L^2(\real^m))}
\end{align}
where $\tilde{a}(x_1) = \norm{a(x_1, \,\cdot\,)}_{L^2(X_2)}$. Clearly, then, $\norm{\tilde{a}}_{L^2(X_1)} = 1$. By a direct application of Proposition \ref{prop:bs} we obtain, for some $C$,
\begin{align*}
\norm{\intt(b \om \tilde{a})}_{S_{q_0,\infty}(L^2(X_1),\, L^2(\real^m))} \le C \esssup_{y \in \real^m}\norm{\om(y, \,\cdot\,)}_{N^s_2(X_1)}. 
\end{align*}
Using the embedding (\ref{eq:sp_incl2}), the above inequality can be used to bound (\ref{eq:int_claim}) to give (\ref{eq:int_ineq}) for some constant $C$.

It therefore remains to prove the claim (\ref{eq:int_claim}). Let $u \in L^2(\real^m)$ and denote $\bx = (x_1,x_2)$. Then straightforward calculations show
\begin{align*}
\big(\intt(b\Om  a)\, \intt(b\Om a)^*u\big)(y) &= \int_{X_1 \times X_2} \int_{\real^m} |a(\bx)|^2 b(y) \overline{b(y')} \Om(y,\bx) \overline{\Om(y',\bx)} u(y')\,dy'\,d\bx\\
&= \int_{X_1} \int_{\real^m} |\tilde{a}(x_1)|^2 b(y) \overline{b(y')} \om(y,x_1) \overline{\om(y',x_1)} u(y')\,dy'\,dx_1\\
&= \big(\intt(b \om \tilde{a}) \,\intt(b \om \tilde{a})^*u\big)(y)
\end{align*}
for $y \in \real^m$. By the definition of singular values,
\begin{align*}
s_k\big(\intt(b\Om a)^*\big) = \lambda_k\big(\intt(b\Om a) \intt(b\Om a)^*\big)^{1/2},
\end{align*}
where $\lambda_k$ denotes the $k$-th positive eigenvalue in a sequence ordered non-increasing and counting multiplicity. Therefore, for all $k$, $s_k(\intt(b\Om a)^*) = s_k(\intt(b\om \tilde{a})^*)$, and so $s_k(\intt(b\Om a)) = s_k(\intt(b\om \tilde{a}))$. This completes the proof of the claim.
\end{proof}

We now apply the preceeding result to show $\mathcal{C}$ of (\ref{eq:c}) can act as a multiplier in suitable $S_{p,\infty}$-spaces. To begin, we consider the function $e^{\tau/4}$ on $\real^3$ for $\tau$ as in (\ref{eq:tau}). By straightforward calculations, we see that the conditions of Lemma \ref{lem:du} are satisfied for this function with $M=1$, $a_1 = 0$, $\al = 1$ and $q=2$. Therefore, for any unit ball $Q \subset \real^3$ we have $e^{\tau/4} \in N^{5/2}_2(Q)$ and
\begin{align}
\label{eq:tau_besov}
\norm{e^{\tau/4}}_{N^{5/2}_2(Q)} \le C_0
\end{align}
for some $C_0$ independent of $Q$.

Since $\mathcal{C}$ is defined as a product of such exponentials, we obtain the following lemma. Here, $1 \le K \le N-1$ and as usual we write $\check\bx = (x_1, \dots, x_K)$, $\hbx = (x_{K+1}, \dots, x_N)$.
For any $\bnu \in \real^{3K}$ set
\begin{align}
\label{eq:qnu}
Q_{\bnu} = [0,1)^{3K}+\bnu.
\end{align}
\begin{lem}
\label{lem:c_multiplier}
Take any $\bnu \in \real^{3K}$. Then $\La(\hbx,\check\bx) = \mathcal{C}(\check\bx, \hbx)$, for $\mathcal{C}$ defined in (\ref{eq:c}), is a multiplier on $S_{q,\infty} = S_{q,\infty}(L^2(Q_{\bnu}), L^2(\real^{3(N-K)}))$ for any $3/4<q <2$ and
\begin{align*}
M_{q,\infty}(\La) \le C
\end{align*}
for some $C$ potentially depending on $q$ and $N$, but independent of $\bnu$.
\end{lem}
\begin{proof}
First, we write $\La(\hbx,\check\bx) = \prod_{l=1}^K\prod_{m=K+1}^N \La_{l,m}(\hbx,\check\bx)$ for $\La_{l,m}(\hbx,\check\bx) = \exp(\tau(x_l - x_m)/4)$.  It suffices to consider individual $\La_{l,m}$ since clearly products of multipliers are multipliers in the same space. Therefore, take some $l \in \{1,\dots,K\}$ and $m \in \{K+1, \dots, N\}$. Next, for $\bnu = (\nu_1, \dots, \nu_K)$ considered fixed, we write $Q_j = [0,1)^3 + \nu_j$ for each $j=1,\dots,K$, and so $Q_{\bnu} = Q_1 \times \dots \times Q_K$. By (\ref{eq:tau_besov}) we have
\begin{align*}
\sup_{x_m \in \real^3}\norm{\exp(\tau(\,\cdot\,-x_m)/4)}_{N^{5/2}_2(Q_l)} \le C_0.
\end{align*}
Applying Lemma \ref{lem:multiplier} with $X_1 = Q_l$ and $X_2 = Q_1 \times Q_{l-1} \times Q_{l+1} \times Q_K$ we obtain the required result for $\La_{l,m}$.

\end{proof}

\section{Proof of Proposition \ref{prop:psi_k}}
\label{chpt:psi_k}
As in (\ref{eq:qnu}), we set $Q_{\bn} = [0,1)^{3K} + \bn$ for $\bn \in \integers^{3K}$. The notation $\intt(\,\cdot\,)$ was defined in (\ref{eq:io_def}). We start with a general proposition which brings together results of the previous sections.
\begin{prop}
\label{prop:general}
Let $1 \le K \le N-1$. For $\al >0$, let $u \in C^{\infty}(\real^{3N}\backslash\Sigma)$ obey (\ref{eq:mb_pointwise}) for some $A>0$. Denote $\mathtt{U}(\hbx, \check\bx) = u(\check\bx, \hbx)$ for $\check\bx \in \real^{3K}$ and $\hbx \in \real^{3(N-K)}$. We set
\begin{align}
\label{eq:final_q}
\frac{1}{q} = \frac{1}{2} + \frac{2\al + 3}{6K}.
\end{align} 
Furthermore, take any $r =2$ if $2\al>3(K-1)$ and $r \in (6K/(2\al + 3),\, \infty]$ otherwise. Let $a = a(\check\bx)$ lie in $L_{loc}^r(\real^{3K})$ and $b = b(\hbx)$ lie in $L^2(\real^{3(N-K)})$. Also, suppose $\La = \La(\hbx, \check\bx) \in L^{\infty}(\real^{3N})$ is a multiplier on $S_{q,\infty}(L^2(Q_{\bn}),\, L^2(\real^{3(N-K)}))$ for each $\bn \in \integers^{3K}$. Denote $L_{\bn} := M_{q,\infty}(\La)$ for each $\bn \in \integers^{3K}$. 
Then $\textnormal{Int}(b \La\mathtt{U} a) \in S_{q,\infty}(L^2(\real^{3K}), L^2(\real^{3(N-K)}))$ and there exists some finite $C$ such that
\begin{align}
\label{eq:general}
\norm{\intt(b \La \mathtt{U} a)}_{q,\infty} \le CA \norm{b}_{L^2(\real^{3(N-K)})} \Big(\sum_{\bn \in \integers^{3K}} L^q_{\bn} \norm{a}^q_{L^r(Q_{\bn})}\Big)^{1/q},
\end{align}
under the assumption the right-hand side is finite.
\end{prop}
\begin{proof}
Take any $\bn \in \integers^{3K}$. By Lemma \ref{lem:pointwise} we get
\begin{align}
\label{eq:u_besov}
\norm{u(\,\cdot\,, \hbx)}_{N^{\al+3/2}_2(Q_{\bn})} \le CA
\end{align}
for a.e. $\hbx \in \real^{3(N-K)}$. Here, $C$ is independent of $\bn$.

We now apply Proposition \ref{prop:bs} with $d=3K$, $m=3(N-K)$, $s = \al+3/2$ and $X = Q_{\bn}$. This gives that $\intt(b\mathtt{U}a) : L^2(Q_{\bn}) \to L^2(\real^{3(N-K)})$ lies in $S_{q,\infty}$ for $q$ as in (\ref{eq:final_q}) and, due to (\ref{eq:u_besov}), we get
\begin{align*}
\norm{\intt(b\mathtt{U}a)}_{q,\infty} \le CA \norm{b}_{L^2(\real^{3(N-K)})} \norm{a}_{L^r(Q_{\bn})}
\end{align*}
for some constant $C$ independent of $\bn$. The following is then immediate by the definition of multipliers,
\begin{align*}
\norm{\intt(b\La\mathtt{U}a)}_{q,\infty} \le CA L_{\bn} \norm{b}_{L^2(\real^{3(N-K)})} \norm{a}_{L^r(Q_{\bn})}.
\end{align*}

For each $\bn$ we can define $T_{\bn} = \intt(b\La\mathtt{U}a\mathds{1}_{Q_{\bn}})$, for $\mathds{1}_{Q_{\bn}} = \mathds{1}_{Q_{\bn}}(\check\bx)$, which we consider as an operator $L^2(\real^{3K}) \to L^2(\real^{3(N-K)})$. Then $\intt(b\La\mathtt{U}a)$ as an operator $L^2(\real^{3K}) \to L^2(\real^{3(N-K)})$ is the operator $\sum_{\bn \in \integers^{3K}}T_{\bn}$. We note that $T_{\bm}^{}T_{\bn}^* = 0$ whenever $\bm \ne \bn$. By Proposition \ref{prop:triangle} for operators $T_{\bn}$, as above, we obtain precisely (\ref{eq:general}).
\end{proof}

The proof of Proposition \ref{prop:psi_k} is then a straightforward application of the above proposition.
\begin{proof}[Proof of Proposition \ref{prop:psi_k}]
Let $\mathcal{A}(\check\bx)$, $\mathcal{B}(\hbx)$ and $\mathcal{C}(\check\bx, \hbx)$ be the functions defined in (\ref{eq:a})-(\ref{eq:c}) and let $\mu$ be the function defined in (\ref{eq:mu_def}). We then set $\mathtt{U}(\hbx, \check\bx) = \mu(\check\bx, \hbx)$ and $\La(\hbx, \check\bx) = \mathcal{C}(\check\bx, \hbx)$, where $\check\bx \in \real^{3K}$ and $\hbx \in \real^{3(N-K)}$. In view of (\ref{eq:psi_decomp}), we have
\begin{align*}
\Psi^{(K)} = \intt(\mathcal{B} \La \mathtt{U} \mathcal{A}).
\end{align*}

We apply now Proposition \ref{prop:general}. The function $u = \mu$ obeys the assumptions of this proposition for $\al = 2$, by (\ref{eq:mu}). We take $a = \mathcal{A}$, $b = \mathcal{B}$ and any $r$ as required by the proposition. We take $\La$ as above, which is an appropriate multiplier due to Lemma \ref{lem:c_multiplier} (where $\La$ is defined identically). There then exists some $C$ such that
\begin{align*}
\norm{\Psi^{(K)}}_{q,\infty} \le C \norm{\mathcal{B}}_{L^2(\real^{3(N-K)})} \Big(\sum_{\bn \in \integers^{3K}} \norm{\mathcal{A}}^q_{L^{r}(Q_{\bn})}\Big)^{1/q}
\end{align*}
which is a finite quantity.
\end{proof}

\appendix
\section{Proof of Proposition \ref{prop:phi_bound}}
\label{chpt:app}
We prove Proposition \ref{prop:phi_bound} in two steps. Firstly for derivatives up to order two, then for all remaining orders.

\subsection{Derivative of $\phi$ up to order two}
This follows from a result by S. Fournais, M. and T. Hoffmann-Ostenhof and T.\O. S\o rensen in \cite{fhos05}, which we now state.

Take some $\chi \in C_c^{\infty}(\real)$ with $0 \le \chi \le 1$ such that $\chi(x) = 1$ for $|x| \le 1$ and $\chi(x) = 0$ for $|x| \ge 2$. Then we can define on $\real^{3N}$,
\begin{align}
\label{eq:phio}
\phi_{opt} = e^{-G-H}\psi = e^{F-G-H}\phi,
\end{align}
where $F$ and $\phi$ are (\ref{eq:f_def}) and (\ref{eq:phi_def}), respectively, and
\begin{align}
G(\bx) = -\frac{Z}{2}\sum_{1 \le j \le N}\chi(|x_j|)|x_j| + \frac{1}{4}\sum_{1\le j<k \le N}\chi(|x_j-x_k|)|x_j-x_k|,
\end{align}
\begin{align}
H(\bx) = C_0 Z \sum_{1 \le j<k \le N} \chi(|x_j|)\chi(|x_k|) (x_j \cdot x_k)\log(|x_j|^2 + |x_k|^2),
\end{align}
where $C_0 = (2-\pi)/(12\pi)$.
\begin{thm}\cite[Theorem 1.5/Remark 1.6]{fhos05}
\label{thm:phio}
For all $0<r<R$ there exists a constant $C$ such that
\begin{align*}
\norm{\phi_{opt}}_{W^{2,\infty}(B(\bx, r))} \le C\norm{\phi_{opt}}_{L^{\infty}(B(\bx, R))}\quad \text{for all} \quad \bx \in \real^{3N}.
\end{align*}
\end{thm}

Now, $F, G, H \in L^{\infty}(\real^{3N})$, and we have $\partial(G-F) \in L^{\infty}(\real^{3N})$ where $\partial$ represents any partial derivative in $\real^{3N}$ of any order. In addition, for any $\al \in \naturals_0^3$, $|\al| \le 2$, we have $\partial_{x_j}^{\al} H \in L^{\infty}(\real^{3N})$, for $j=1,\dots,N$. Thus, applying the Leibniz rule on (\ref{eq:phio}), using Theorem \ref{thm:phio}, we obtain the following. For any $|\al| \le 2$, $1\le j \le N$ and $0<r<R$ there exists some $C$ such that
\begin{align}
\label{eq:phi_2nd}
\norm{\partial^{\al}_{x_j}\phi}_{L^{\infty}(B(\bx, r))} \le C\norm{\psi}_{L^{\infty}(B(\bx, R))}
\end{align}
for all $\bx \in \real^{3N}$. For any choice of $R$, for example $R=1$, we can then use the above bound along with (\ref{eq:exp_decay}) to obtain (\ref{eq:phi_bound}) in the case of $|m| \le 2$.

\subsection{Higher order derivatives of $\phi$}
In this section, we summarise theory present in \cite[Section 3]{hs23}, adapted and simplified for our purposes. This will be used to complete the proof of Proposition \ref{prop:phi_bound}. We start by considering the equation
\begin{align}
\label{eq:general_pde}
(-\Delta + \ba(\bx) \cdot \nabla + b(\bx))u = 0
\end{align}
in $\real^{3N}$ where the coefficients $\ba : \real^{3N} \to \complex^{3N}$ and $b : \real^{3N} \to \complex$ are locally bounded.

We will consider $\bx_0 \in \real^{3N}$, $l \in (0,1]$, $R>0$ and $j \in \{1,\dots,N\}$ for which these coefficients obey the following condition. For each $m \in \naturals_0^3$ we can find a $C$ such that
\begin{align}
\label{eq:ab_condition}
\norm{\partial_{x_j}^m\ba}_{L^{\infty}(B(\bx_0,\, Rl))} + \norm{\partial_{x_j}^m b}_{L^{\infty}(B(\bx_0,\, Rl))} \le Cl^{-|m|}
\end{align}
where $C$ may depend on $R$. However, when we will consider more than one value of either $\bx_0$ or $l$ we require $C$ to be independent of both.

The following result is \cite[Theorem 3.2]{hs23}, adapted for our purposes. The original statement uses so-called cluster sets, however we do not need this, and we only apply the result to single clusters of the form of singletons i.e. $\{j\}$ for $j=1,\dots,N$.
\begin{prop}
\label{prop:u}
Let $u \in W_{loc}^{1,2}(\real^{3N})$ be a weak solution to (\ref{eq:general_pde}) in $\real^{3N}$. Take some $R>0$, $j \in \{1,\dots,N\} $, $\bx_0 \in \real^{3N}$ and $l \in (0,1]$ which obey the condition (\ref{eq:ab_condition}). Then, for each $m \in \naturals_0^3$, $|m| \ge 2$, and all $r < R$ the weak derivative $\partial_{x_j}^m u$ lies in $C^1(\overline{B(\bx_0, r)})$ and there exists a constant $C$ such that
\begin{align}
\label{eq:u_bound}
\norm{\partial_{x_j}^m u}_{L^{\infty}(B(\bx_0,\, rl))} \le Cl^{2-|m|}\Big(\max_{|\al| \le 2}\norm{\partial_{x_j}^{\al}u}_{L^{\infty}(B(\bx_0, Rl))} + \norm{\nabla u}_{L^{\infty}(B(\bx_0, Rl))} + \norm{u}_{L^{\infty}(B(\bx_0, Rl))} \Big)
\end{align}
where the constants depend on $m$, $r$ and $R$, but do not depend on values of $\bx_0$ and $l$ for which the same constants in (\ref{eq:ab_condition}) may be used.
\end{prop}

Since $\psi$ obeys $H\psi = E\psi$ in $\real^{3N}$, it can be shown by standard computations that $\phi$ is a weak solution to
\begin{align*}
-\Delta\phi - 2\nabla F\cdot\nabla\phi + (V - \Delta F - |\nabla F|^2 - E)\phi = 0
\end{align*}
in $\real^{3N}$ where $V$ is as in (\ref{eq:v}) and $F$ is as in (\ref{eq:f_def}). This equation can then be written as (\ref{eq:general_pde}) for $\ba = - 2\nabla F$ and $b = V - \Delta F - |\nabla F|^2 - E$. We need to show these coefficients are such that the above proposition can be used. Firstly, by evaluation of $\Delta F$ it can readily be shown that
\begin{align}
\label{eq:fv_bounded}
\nabla^k(-\Delta F + V) \in L^{\infty}(\real^{3N}), \quad k=0,1,\dots
\end{align}
We will require the following distance functions. For each $j =1,\dots,N$,
\begin{align}
\label{eq:lam}
\lambda_j(\bx) = \min\big\{1,\, |x_j|,\, 2^{-1/2}|x_j - x_k| : 1\le k \le N,\, k \ne j \big\}.
\end{align}
for $\bx \in \real^{3N}$. We note that if $\bx \in \real^{3N}\backslash\Sigma$ then $\lambda_j(\bx)>0$.

We proceed to apply the above proposition to prove Proposition \ref{prop:phi_bound} for $|m| \ge 2$. Let $0<r<R<1$ and take any $\bx \in \real^{3N}\backslash\Sigma$. By (\ref{eq:fv_bounded}) and \cite[Corollary 4.5]{hs23}, the condition (\ref{eq:ab_condition}) holds for $\bx_0 = \bx$ and $l = \lambda_j(\bx)$, where the constants are independent of the choice of $\bx$. Therefore, by using Proposition \ref{prop:u}, we get for each $|m| \ge 2$ some $C$ such that
\begin{multline*}
\norm{\partial_{x_j}^m \phi}_{L^{\infty}(B(\bx,\, r\lambda_j(\bx)))} \le C\lambda_j(\bx)^{2-|m|}\Big(\max_{|\al| \le 2}\norm{\partial_{x_j}^{\al}\phi}_{L^{\infty}(B(\bx, R\lambda_j(\bx)))}\\ + \norm{\nabla \phi}_{L^{\infty}(B(\bx, R\lambda_j(\bx)))} + \norm{\phi}_{L^{\infty}(B(\bx, R\lambda_j(\bx)))} \Big)
\end{multline*}
for all $\bx \in \real^{3N}\backslash\Sigma$. We simplify the left-hand side, and use (\ref{eq:lam}) and (\ref{eq:phi_2nd}) to write
\begin{align*}
|\partial_{x_j}^m \phi(\bx)| \le C\Big(1+ |x_j|^{2-|m|} + \sum_{k \ne j}|x_j-x_k|^{2-|m|} \Big) \norm{\psi}_{L^{\infty}(B(\bx, 1))}
\end{align*}
for some new constants. This produces the required result (\ref{eq:phi_bound}) for all $|m| \ge 2$, after applying (\ref{eq:exp_decay}).

\vskip 0.5cm
\noindent
\textbf{Acknowledgments.} The author is grateful to A.V. Sobolev and S. Fournais for helpful discussions.
\vskip 0.5cm
\noindent
\textbf{Funding.} This work was partially supported by the Villum Centre of Excellence
for the Mathematics of Quantum Theory (QMATH) with Grant No.10059. The author would also like to acknowledge previous support by Engineering and Physical Sciences Research Council (EPSRC) grant EP/W522636/1.

\bibliographystyle{unsrt}
\bibliography{refs}

\end{document}